\documentclass[onecolumn, a4size, 12pt]{IEEEtran}
\usepackage{amsmath}
\usepackage{amssymb}
\usepackage{amsfonts}
\usepackage{graphicx}
\usepackage{epsfig}
\usepackage{subfigure}
\usepackage{psfrag}

\title{On Active Learning and Supervised Transmission of Spectrum Sharing Based Cognitive
Radios by Exploiting Hidden Primary Radio Feedback\footnote{This
paper has been presented in part at IEEE Global Communications
Conference (Globecom), December 2009.}\footnote{R. Zhang is with the
Institute for Infocomm Research, A*STAR, Singapore and the
Department of Electrical and Computer Engineering, National
University of Singapore (e-mail:rzhang@i2r.a-star.edu.sg).}}

\author{Rui Zhang}

\begin{document}
\maketitle \thispagestyle{empty} \vspace{-0.5in}
\begin{abstract}
This paper studies the wireless spectrum sharing between a pair of
distributed primary radio (PR) and cognitive radio (CR) links.
Assuming that the PR link adapts its transmit power and/or rate upon
receiving an interference signal from the CR and such transmit
adaptations are observable by the CR, this results in a new form of
feedback from the PR to CR, refereed to as \emph{hidden PR
feedback}, whereby the CR learns the PR's strategy for transmit
adaptations without the need of a dedicated feedback channel from
the PR. In this paper, we exploit the hidden PR feedback to design
new learning and transmission schemes for spectrum sharing based
CRs, namely \emph{active learning} and \emph{supervised
transmission}. For active learning, the CR initiatively sends a
probing signal to interfere with the PR, and from the observed PR
transmit adaptations the CR estimates the channel gain from its
transmitter to the PR receiver, which is essential for the CR to
control its interference to the PR during the subsequent data
transmission. This paper proposes a new transmission protocol for
the CR to implement the active learning and the solutions to deal
with various practical issues for implementation, such as time
synchronization, rate estimation granularity, power measurement
noise, and channel variation. Furthermore, with the acquired
knowledge from active learning, the CR designs a {\it supervised}
data transmission by effectively controlling the interference powers
both to and from the PR, so as to achieve the optimum performance
tradeoffs for the PR and CR links. Numerical results are provided to
evaluate the effectiveness of the proposed schemes for CRs under
different system setups.
\end{abstract}


\setlength{\baselineskip}{1.3\baselineskip}

\newtheorem{definition}{Definition}
\newtheorem{assumption}{Assumption}
\newtheorem{theorem}{\underline{Theorem}}[section]
\newtheorem{lemma}{\underline{Lemma}}[section]
\newtheorem{corollary}{Corollary}
\newtheorem{proposition}{\underline{Proposition}}[section]
\newtheorem{example}{\underline{Example}}[section]
\newtheorem{remark}{\underline{Remark}}[section]
\newtheorem{algorithm}{\underline{Algorithm}}[section]

\section{Introduction}

Opportunistic spectrum access (OSA) and spectrum sharing (SS) are
two basic operation models for the secondary radio or so-called
cognitive radio (CR) system to operate over a common frequency band
with an existing primary radio (PR) system. For the OSA model (see,
e.g., \cite{Haykin05}), the CR usually deploys a spectrum sensing
technique to detect the PR transmission on-off status over the
frequency band of interest, and decides to transmit over this band
if the sensing result indicates that the PR is not transmitting with
a high probability. In contrast, the SS model (see, e.g.,
\cite{Tarokh06,Viswanath09,Zhang10}) allows the CR to transmit
concurrently with the PR over the same frequency band, provided that
the CR knows how to control its interference to the PR such that the
resultant PR performance degradation is tolerable. Since SS-based
CRs in general utilize the spectrum more efficiently than OSA-based
CRs, this paper focuses on the SS model for CRs.

One commonly adopted method for SS-based CRs to protect the PR
transmission is via imposing an {\it interference temperature}
constraint (ITC) over the CR transmission, i.e., the CR interference
power level at each PR receiver must be kept below a prescribed
threshold \cite{Gastpar07a,Ghasemi07,Aissa09,Kang08}. Some important
design issues related to the ITC-based approach are discussed as
follows. First, the effectiveness of the ITC to protect the PR
transmission needs to be addressed. In \cite{Zhang08a} and
\cite{Zhang08b}, it has been shown that the ITC guarantees an upper
bound on the maximum capacity loss of the PR channel due to the CR
interference. In \cite{Zhang09a}, an interesting {\it interference
diversity} phenomenon was discovered, where the average ITC over
different fading states was shown to be superior over the peak ITC
counterpart for minimizing the PR ergodic/outage capacity losses.
Second, it is pertinent to investigate more efficient methods for
the CR to protect the PR than that with a fixed ITC. Such methods
may exploit additional side information on the PR transmissions such
as the PR's on-off status \cite{Zhang08b}, Automatic Repeat reQuest
(ARQ) feedback \cite{Gastpar07b}, channel state information (CSI)
\cite{Zhang08b,Chen07}, spatial signal space
\cite{Zhang08a,Zhang08d}, and frequency power allocation
\cite{WenyiZhang}, in order to set more appropriate interference
power levels over time, frequency, or space for CR's opportunistic
transmission. Thus, conventional ITCs are replaced by the more
relevant {\it PR performance loss} constraints
\cite{Zhang08b,Kang09}. However, although these new methods are
promising to improve the PR and CR spectrum sharing throughput, they
usually require substantial overheads for implementation as compared
with the ITC. Third, even implementation of the ITC requires
knowledge of the channel gain from the CR transmitter to the PR
receiver, which is difficult to obtain for the CR without a
dedicated feedback channel from the PR. If the PR link adopts a
time-division-duplex (TDD) mode and thus the channel reciprocity
holds between PR and CR terminals, the CR-to-PR channel gain can
then be estimated by the CR from its observed PR signals, assuming
prior knowledge of the PR transmit power. However, if a
frequency-division-duplex (FDD) mode is adopted by the PR (i.e., PR
terminal transmits and receives over two different frequency bands),
channel reciprocity between PR and CR terminals does not hold in
general. As a result, estimating CR-to-PR channels from the observed
PR signals may fail for the CR.

Motivated by the above discussions, this paper presents a new design
paradigm for SS-based CRs, which resolves the CR-to-PR channel
estimation problem for the CR, and also leads to a more efficient
spectrum sharing solution than the conventional one with fixed ITCs.
The proposed method exploits an interesting PR-CR interaction by
assuming that the PR deploys certain form of transmit power and/or
rate adaptations upon receiving an interference signal from the
CR.\footnote{Under this assumption, this paper considers PR systems
that have two-way communications such that one node can send control
signals to the other node for transmit adaptation. Such PR systems
apparently do not apply to one-way communication systems (e.g., the
TV broadcasting system considered for WRAN \cite{802.22}), but may
find applications in existing cellular-based wireless systems (see,
e.g., \cite{Femtocell}).} Specifically, suppose that the CR
initially transmits a probing signal to interfere with the PR
receiver, which then sends back a control signal (via the PR
feedback channel) to the PR transmitter for adapting transmit power
and/or rate accordingly; finally, the PR transmit adaptations are
observed by the CR. Thereby, the CR obtains knowledge on the PR
deployed strategy for transmit adaptations without the need of a
dedicated feedback channel from the PR. This implicit form of
feedback from the PR to CR is thus named as {\it hidden PR
feedback}. Since the CR initiatively sends a probing signal to
interfere with the PR for activating the hidden PR feedback, this
``active learning'' principle is different from existing ``passive
learning'' counterpart (e.g., detecting the PR on-off status or
estimating the CR-to-PR channel gain via sensing the PR band only)
for the design of CR systems. However, it should be pointed out that
the probing signal from the CR can cause a temporary performance
degradation of the PR, and thus needs to be properly designed
(details will be given later in the paper). The use of active
learning approach for designing new spectrum sensing techniques for
OSA-based CRs have been studied in \cite{Zhang08c} and \cite{Li09},
while in this paper we apply this interesting approach to design new
learning and transmission schemes for SS-based CRs. It is worth
noting that although iteratively adapting transmit power and rate to
cope with the co-channel interference among users in decentralized
communication systems has been studied in the literature (see, e.g.,
\cite{Foschini,Yates,Yu02}), the approach of exploiting the PR
transmit adaptations to design new operation schemes for the CR is a
new contribution of this paper. Based on the hidden PR feedback,
this paper proposes two new types of operations for SS-based CRs,
which are described as follows.

\begin{itemize}
\item {\it Active Learning}: By probing the PR with interference and observing its transmit power/rate adaptations,
under certain conditions, the CR is able to estimate the channel
gain from its transmitter to the PR receiver, which is essential for
the CR to control its interference to the PR during subsequent data
transmission. We refer to this new scheme for the CR as {\it active
learning}, to differ it from existing passive learning schemes in
the literature.

\item {\it Supervised Transmission}: With the acquired knowledge on
the CR-to-PR channel gain and the PR transmit adaptations from
active learning, the CR is able to design a {\it supervised data
transmission} via controlling the interference power levels both to
and from the PR. Thus, the CR ensures that the resultant performance
degradation of the PR is within a tolerable margin, and the CR
achievable rate is optimized under the ``feedback'' interference
from the PR, which is in general coupled with the CR transmit power
due to the CR-to-PR interference and the resultant PR power
adaptation.
\end{itemize}

This paper proposes a new transmission protocol for the CR to
implement active learning, together with solutions to deal with
various important practical issues such as time discrepancy between
the PR and CR links, CR rate estimation granularity and power
measurement noise, and PR/CR channel variations. This paper also
analyzes the PR and CR jointly achievable rates with the CR
supervised transmission. Moreover, this paper evaluates the
effectiveness of the proposed CR learning and transmission schemes
when the PR employs different transmit power/rate adaptation schemes
over the fading channels \cite{Goldsmith97}.

The rest of this paper is organized as follows. Section
\ref{sec:system model} presents the system model. Section
\ref{sec:PR feedback} describes the hidden PR feedback with
different PR transmit adaptation strategies. Section
\ref{sec:learning} presents the active learning method for the CR to
estimate the CR-to-PR channel gain, a protocol to implement this
method and various solutions to deal with practical issues. Section
\ref{sec:CR transmission} studies the CR supervised data
transmission by analyzing the achievable rates of both the PR and CR
links. Section \ref{sec:numerical examples} provides numerical
examples to corroborate the proposed studies. Finally, Section
\ref{sec:conclusion} concludes the paper.

\section{System Model}\label{sec:system model}

As shown in Fig. \ref{fig:system model}, for the purpose of
exposition, this paper considers a simplified spectrum sharing
system, where one CR link consisting of a CR transmitter (CR-Tx) and
a CR receiver (CR-Rx) shares a narrow-band for transmission with one
PR link consisting of a PR transmitter (PR-Tx) and a PR receiver
(PR-Rx). All the terminals involved are assumed to be each equipped
with a single antenna. We assume a block-fading channel model for
all the channels shown in Fig. \ref{fig:system model}. We also
assume coherent communication for both the PR and CR links and thus
only the fading channel power gain (amplitude square) is of
interest. In addition, since the proposed study in this paper
applies to any particular channel fading state, for notational
brevity, we drop the channel fading state index for the following
definitions. Denote $h_c$, $h_p$, $h_{cp}$, and $h_{pc}$ as the
power gains of the channels from CR-Tx to CR-Rx, from PR-Tx to
PR-Rx, from CR-Tx to PR-Rx, and from PR-Tx to CR-Rx, respectively.
In addition, denote $\tilde{h}_{pc}$ as the channel power gain from
PR-Tx to CR-Tx. Without loss of generality, it is assumed that the
additive noises at both PR-Rx and CR-Rx are independent circularly
symmetric complex Gaussian (CSCG) random variables with zero mean
and variances denoted by $\sigma_p^2$ and $\sigma_c^2$,
respectively.

First, consider the PR link. It is assumed that the PR is oblivious
to the existence of the CR and treats the interference from CR-Tx as
additional noise at the receiver. We assume that the PR employs
certain form of transmit power and/or rate adaptations based upon
the PR CSI as well as the interference power level received from the
CR. Let $N_p$ denote the noise-plus-interference power level at
PR-Rx, i.e., $N_p=\sigma_p^2+h_{cp}p_c$, with $p_c$ denoting the
transmit power of the CR. The PR transmit power, denoted by $p_p$,
is then given by $\mathcal{P}_p(\gamma_p)$, which defines a mapping
from the PR ``effective'' channel power gain, $\gamma_p=h_p/N_p$, to
$p_p$. The PR is assumed to employ packet-based transmissions and
the transmit rate of one particular packet is denoted by $r_p$. For
a given pair of $\gamma_p$ and $p_p$, $r_p$ is assumed equal to
$\mathcal{R}_p(SNR_p)$, with $SNR_p=\gamma_pp_p$ denoting the
signal-to-noise (including both the additive noise and CR
interference) ratio (SNR) at PR-Rx. Note that the rate function
$\mathcal{R}_p(SNR_p)$ is specified by the employed modulation and
coding scheme (MCS) of the PR link.

Next, consider the CR link. The CR is assumed to be aware of the PR,
and furthermore protect the PR transmission by ensuring that the
resultant performance loss of the PR due to the CR interference is
within a tolerable margin. However, we consider a practical scenario
where there is no dedicated communication channel for the PR to send
any side information (e.g., $h_{cp}$) to the CR for facilitating its
interference control to the PR. Consequently, the CR needs to fulfil
the task of protecting the PR by its own effort. In this case, one
possible method for the CR is to deploy spectrum sensing techniques
to detect the PR on-off status, and then transmit if the sensing
result indicates that the PR is not transmitting with a high
probability (i.e., OSA-based CRs). In contrast, this paper studies
more efficient methods for the CR to utilize the PR spectrum than
sensing-based orthogonal transmission, where the CR manages to
transmit even when the PR is transmitting over the same band (i.e.,
SS-based CRs).

\section{Hidden PR Feedback}\label{sec:PR feedback}

In this section, we illustrate the phenomenon of hidden PR feedback.
First, consider for the PR link the following three commonly adopted
power control policies in wireless communication:
\begin{itemize}
\item {\it Constant Power (CP) Policy}:
$\mathcal{P}_p(\gamma_p)=Q, \forall \gamma_p\geq 0$, where $Q$ is a
constant;
\item {\it Persistent Power Control Policy}: $\mathcal{P}_p(\gamma_p^{(2)})\geq
\mathcal{P}_p(\gamma_p^{(1)})$, for any $0<\gamma_p^{(2)}<
\gamma_p^{(1)}$;
\item {\it Non-Persistent Power Control Policy}: $\mathcal{P}_p(\gamma_p^{(2)})\leq
\mathcal{P}_p(\gamma_p^{(1)})$, for any $0<\gamma_p^{(2)}<
\gamma_p^{(1)}$.
\end{itemize}

The CP policy is usually applied when PR-Tx has a strict peak power
constraint given by $Q$ over all transmitted packets, while the
other two policies are applicable when PR-Tx is subject to an
average power constraint and thus can change transmit powers over
different packets. Note that with the {\it persistent} power
control, $p_p$ usually increases when the effective channel power
gain, $\gamma_p$, decreases. This type of power control is usually
applied for data traffic with a stringent quality-of-service (QoS)
requirement in terms of receiver SNR, $SNR_p=\gamma_pp_p$. One
well-known example  in the literature for the persistent power
control is the so-called {\it truncated channel inversion} (TCI)
\cite{Goldsmith97},\footnote{Strictly speaking, TCI is
non-persistent only for the regime of $\gamma_p>\gamma_p^{(T)}$.
Alternatively, TCI is non-persistent for all values of $\gamma_p$ in
the special case of $\gamma_p^{(T)}=0$, where TCI reduces to the
conventional {\it channel inversion} power control
\cite{Goldsmith97}.} which is expressed as
\begin{align}\label{eq:TCI}
p_p^{\rm TCI}=\left\{\begin{array}{ll} \frac{SNR_p^{(T)}}{\gamma_p}
& ~~ {\rm if}~ \gamma_p>\gamma_p^{(T)} \\ 0 & ~~ {\rm otherwise}
\end{array} \right.
\end{align}
where $SNR_p^{(T)}$ is the given SNR target, while $\gamma_p^{(T)}$
is the threshold for $\gamma_p$ below which the PR decides to take a
``transmit outage'', i.e., $p_p=0$ and thus $r_p=0$.
$\gamma_p^{(T)}$ can be determined from the PR average transmit
power constraint and is related to the PR outage probability
\cite{Goldsmith97} (details are omitted here for brevity). With the
TCI power control, the PR transmits with a constant rate
$r_p=\mathcal{R}_p(SNR_p^{(T)})$ if $\gamma_p\geq\gamma_p^{(T)}$.

In contrast, with the {\it non-persistent} power control, the PR
usually decreases its transmit power when $\gamma_p$ decreases, in
order to save transmit powers for better opportunities with larger
values of $\gamma_p$. One well-known example for the non-persistent
power control is the so-called {\it water-filling} (WF)
\cite{Goldsmith97} policy, which is given by
\begin{align}\label{eq:WF}
p_p^{\rm WF}=\left\{\begin{array}{ll} \mu-\frac{1}{\gamma_p} & ~~
{\rm if}~ \gamma_p>\frac{1}{\mu} \\ 0 & ~~ {\rm otherwise}
\end{array} \right.
\end{align}
where $\mu$ is a constant, or the so-called ``water-level'', which
can be determined from the PR average transmit power constraint
\cite{Goldsmith97} (details are omitted here). The WF power control
results in a variable-rate transmission for the PR, where
$r_p=\mathcal{R}_p(\gamma_p\mu-1)$ if $\gamma_p> (1/\mu)$; and
$r_p=0$ otherwise.

From the above discussions, it is observed that $p_p$ and/or $r_p$
may vary with the values of $\gamma_p$. Since
$\gamma_p=h_p/(\sigma_p^2+h_{cp}p_c)$ for a given fading state with
fixed channel power gains $h_p$ and $h_{cp}$, it follows that
$\gamma_p$ is solely determined by transmit power of the CR signal,
$p_c$. More specifically, we can express $p_p$ and $r_p$ in terms of
$p_c$ for CP, TCI, and WF power control of the PR as follows.
\begin{align}
p_p^{\rm CP}&=Q.
\label{eq:power CP} \\
r_p^{\rm
CP}&=\mathcal{R}_p\left(\frac{h_pQ}{\sigma_p^2+h_{cp}p_c}\right).
\label{eq:rate CP}
\end{align}

\begin{align}
p_p^{\rm TCI}&=\left\{\begin{array}{ll}
\frac{SNR_p^{(T)}(\sigma_p^2+h_{cp}p_c)}{h_p} & ~~ {\rm if}~
p_c<\left(\frac{h_p}{\gamma_p^{(T)}}-\sigma_p^2\right)\frac{1}{h_{cp}}
\\ 0 & ~~ {\rm otherwise}.
\end{array} \right.\label{eq:power TCI} \\
r_p^{\rm TCI}&= \left\{\begin{array}{ll} \mathcal{R}_p(SNR_p^{(T)})
& ~~ {\rm if}~
p_c<\left(\frac{h_p}{\gamma_p^{(T)}}-\sigma_p^2\right)\frac{1}{h_{cp}}
\\ 0 & ~~ {\rm otherwise}.\end{array} \right. \label{eq:rate TCI}
\end{align}

\begin{align}
p_p^{\rm WF}&=\left\{\begin{array}{ll}
\mu-\frac{\sigma_p^2+h_{cp}p_c}{h_p} & ~~ {\rm if}~ p_c<\frac{\mu
h_p-\sigma_p^2}{h_{cp}}
\\ 0 & ~~ {\rm otherwise}.
\end{array} \right.\label{eq:power WF} \\
r_p^{\rm WF}&= \left\{\begin{array}{ll} \mathcal{R}_p(\frac{\mu
h_p}{\sigma_p^2+h_{cp}p_c}-1) & ~~ {\rm if}~ p_c<\frac{\mu
h_p-\sigma_p^2}{h_{cp}}
\\ 0 & ~~ {\rm otherwise}.\end{array} \right. \label{eq:rate WF}
\end{align}

In Fig. \ref{fig:PR rate power}, $p_p$ and $r_p$ are plotted as
functions of $p_c$, for the CP, TCI (assuming
$h_p>\sigma_p^2\gamma_p^{(T)}$), and WF (assuming
$h_p>\sigma_p^2/\mu$) power control of the PR, respectively. For the
purpose of illustration, in this example we assume that
$\mathcal{R}_p(SNR_p)=\log_2(1+SNR_p)$, which holds when the optimal
Gaussian codebook is used by the PR with interference from the CR
treated as additive Gaussian noise. As observed, by interfering with
the PR with $p_c>0$, the CR is usually able to make the PR change
its transmit power and/or rate for all considered PR power control
policies. As a result, the corresponding changes occur in the
received PR signal power, $\tilde{h}_{pc}p_p$, and/or rate, $r_p$,
at CR-Tx. Therefore, there exists a {\it hidden} PR power and/or
rate feedback observable by the CR, which is activated by the CR via
initiatively interfering with the PR. In the following, we will
apply this hidden PR feedback phenomenon to design new learning and
transmission schemes for the CR.

\section{Active Learning}\label{sec:learning}

In this section, we apply the hidden PR feedback to design CR active
learning with the goal of estimating the channel power gain from
CR-Tx to PR-Rx, $h_{cp}$, which is essential for the CR to control
the interference to the PR during data transmission as discussed
later in Section \ref{sec:CR transmission}. First, we present the
proposed scheme for the ideal case with a number of assumptions
made. Then, we present a protocol for the CR to implement the
proposed scheme and the solutions to deal with important issues for
implementation with relaxed assumptions.

\subsection{CR-to-PR Channel Gain Estimation}

In this subsection, we propose a new scheme for CR-Tx to estimate
$h_{cp}$ via active learning (i.e., without the need of a feedback
channel from PR-Rx) under certain assumptions listed as follows.
\begin{itemize}
\item The CR knows the PR transmission protocol and is able to synchronize its operation with the PR transmission.

\item In the case where the CR needs to extract rate information from the received PR
signal, this can be done by the CR via certain techniques.
Furthermore, the PR transmit rate, $\mathcal{R}_p(SNR_p)$, is a
continuously increasing function of the receiver SNR, $SNR_p$, and
this function is known to the CR.

\item In the case where the CR needs to estimate the received signal
power from the PR, the effect of the receiver noise on the power
estimation is ignored.

\item During the period for the proposed scheme to be implemented, all
the channels involved in Fig. \ref{fig:system model} remain
constant.
\end{itemize}
The above assumptions will be relaxed in the next subsection where
implementation issues for the proposed scheme are addressed.

Next, we present the scheme to estimate $h_{cp}$ as follows. Suppose
that initially CR-Tx listens to the PR transmission,\footnote{In
practice, either CR-Tx or CR-Rx can observe the signal power and/or
rate from PR-Tx to estimate $h_{cp}$ using the method presented in
this paper, while the one between them that has a superior channel
quality from PR-Tx is more suitable for this task. For simplicity,
this paper assumes that this task is done by CR-Tx.}  and observes
the received signal power and rate from PR-Tx, represented by
$q_p^{(0)}=\tilde{h}_{pc}p_p^{(0)}$ and
$r_p^{(0)}=\mathcal{R}_p(\gamma_p^{(0)}p_p^{(0)})$, respectively,
with $p_p^{(0)}$ denoting the initial transmit power of the PR and
$\gamma_p^{(0)}=h_p/\sigma_p^2$. Next, CR-Tx broadcasts a probing
signal of power $p_c$, and PR-Rx reacts upon receiving the
interference from CR-Tx by sending back to PR-Tx (via a dedicated
feedback channel for the PR link) a control signal to indicate
transmit power and/or rate adaptation. Accordingly, PR-Tx resets
transmit power and rate to be $p_p^{(1)}$ and $r_p^{(1)}$,
respectively, where $p_p^{(1)}$ depends on the employed power
control policy $\mathcal{P}_p$ of the PR and
$r_p^{(1)}=\mathcal{R}_p(\gamma_p^{(1)}p_p^{(1)})$ with
$\gamma_p^{(1)}=h_p/(\sigma_p^2+p_ch_{cp})$. As a result, CR-Tx
observes the updated power received from PR-Tx,
$q_p^{(1)}=\tilde{h}_{pc}p_p^{(1)}$, and the updated transmit rate
of the PR, $r_p^{(1)}$. Under the aforementioned assumptions,
$q_p^{(0)}$, $r_p^{(0)}$, $q_p^{(1)}$, and $r_p^{(1)}$ are all
perfectly observed by CR-Tx.

Without loss of generality, it can be assumed that in the above
proposed scheme, $p_p^{(0)}> 0$ and thus $q_p^{(0)}> 0$. This is so
because if $p_p^{(0)}=0$, the PR does not transmit initially, and
thus the CR can simply transmit as if the PR is not present and the
estimation of $h_{cp}$ becomes unnecessary in this case.
Furthermore, note that if $p_p^{(0)}> 0$, there always exists a
non-trivial interval of $p_c$ for which $p_p^{(1)}> 0$. This is
obvious with e.g., CP policy of the PR since $p_p^{(1)}=Q$
regardless of $p_c$, while with TCI power control, from
(\ref{eq:TCI}) it follows that $p_p^{(0)}> 0$ implies that
$\frac{h_p}{\gamma_p^{(T)}}>\sigma_p^2$ and thus $p_p^{(1)}>0$
provided that $p_c<(\frac{h_p}{\gamma_p^{(T)}}-\sigma_p^2)/h_{cp}$;
and with WF power control, from (\ref{eq:WF}) it follows that
$p_p^{(0)}>0$ implies that $\mu h_p>\sigma_p^2$ and thus
$p_p^{(1)}>0$ provided that $p_c<\frac{\mu h_p-\sigma_p^2}{h_{cp}}$.
Thus, without loss of generality, we can also assume that
$q_p^{(1)}>0$ (if not, the CR can re-probe the PR with a smaller
power $p_c$). Consequently, $r_p^{(0)}>0$ and $r_p^{(1)}>0$.

Note that the observed $r_p^{(1)}$ contains side information on
$h_{cp}$ to be estimated via the term $\gamma_p^{(1)}$. However,
$h_{cp}$ cannot be determined solely from $r_p^{(1)}$ since other
relevant terms, $h_p$, $\sigma_p^2$, and $p_p^{(1)}$ are unknown to
the CR. Interestingly, CR-Tx can determine $h_{cp}/\sigma_p^2$ from
the observed $q_p^{(0)}$, $r_p^{(0)}$, $q_p^{(1)}$ and $r_p^{(1)}$,
and the probing signal power $p_c$, as shown in the following
proposition.
\begin{proposition}\label{proposition:channel estimation}
Assuming that $q_p^{(0)}$, $r_p^{(0)}$, $q_p^{(1)}$, and $r_p^{(1)}$
are all strictly positive, the channel power gain from CR-Tx to
PR-Rx $h_{cp}$ normalized to the noise power at PR-Rx $\sigma_p^2$
can be estimated as
\begin{equation}\label{eq:channel gain estimate}
\frac{h_{cp}}{\sigma_p^2}=\left(\frac{\mathcal{R}^{-1}_p(r_p^{(0)})q_p^{(1)}}{\mathcal{R}^{-1}_p(r_p^{(1)})q_p^{(0)}}-1\right)
\frac{1}{p_c}
\end{equation}
where $\mathcal{R}^{-1}_p(\cdot)$ denotes the inverse function of
$\mathcal{R}_p(\cdot)$.
\end{proposition}
\begin{proof}
Since
\begin{equation}\label{eq:equality 1}
\frac{q_p^{(0)}}{q_p^{(1)}}=\frac{\tilde{h}_{pc}p_p^{(0)}}{\tilde{h}_{pc}p_p^{(1)}}=\frac{p_p^{(0)}}{p_p^{(1)}}
\end{equation}
and from the expressions of $r_p^{(0)}$ and $r_p^{(1)}$, it follows
that
\begin{align}
\frac{p_p^{(0)}}{p_p^{(1)}}&=\frac{\mathcal{R}^{-1}_p(r_p^{(0)})\gamma_p^{(1)}}{\mathcal{R}^{-1}_p(r_p^{(1)})\gamma_p^{(0)}}
\label{eq:equality 3}
\\ &=\frac{\mathcal{R}^{-1}_p(r_p^{(0)})\frac{h_p}{\sigma_p^2+p_ch_{cp}}}{\mathcal{R}^{-1}_p(r_p^{(1)}){\frac{h_p}{\sigma_p^2}}} \\
&=\frac{\mathcal{R}^{-1}_p(r_p^{(0)})}{\mathcal{R}^{-1}_p(r_p^{(1)})(1+\frac{p_ch_{cp}}{\sigma_p^2})}.
\label{eq:equality 2}
\end{align}
Using (\ref{eq:equality 1}) and (\ref{eq:equality 2}),
(\ref{eq:channel gain estimate}) can be obtained.
\end{proof}

We see that Proposition \ref{proposition:channel estimation} is
mainly based upon the ``hidden'' equation in (\ref{eq:equality 3}),
which is due to the PR transmit self-adaptation upon receiving the
interference from the CR. Note that the method given in Proposition
\ref{proposition:channel estimation} applies to any general PR
transmit power/rate adaptation strategy, provided that at least one
of the PR transmit power and rate is changed after receiving
interference from the CR. In the two special cases of CP and TCI
power control policies for the PR, for which
$q_p^{(1)}=q_p^{(0)}=\tilde{h}_{pc}Q$ and
$r_p^{(1)}=r_p^{(0)}=\mathcal{R}_p(SNR_p^{(T)})$, respectively, it
easily follows that the estimation rule in (\ref{eq:channel gain
estimate}) reduces to
\begin{align}
\frac{h_{cp}^{\rm
CP}}{\sigma_p^2}&=\left(\frac{\mathcal{R}^{-1}_p(r_p^{(0)})}{\mathcal{R}^{-1}_p(r_p^{(1)})}-1\right)
\frac{1}{p_c}\label{eq:channel gain estimate CP} \\
\frac{h_{cp}^{\rm
TCI}}{\sigma_p^2}&=\left(\frac{q_p^{(1)}}{q_p^{(0)}}-1\right)\frac{1}{p_c}.\label{eq:channel
gain estimate TCI}
\end{align}
Therefore, only rate/power adaptation of the PR needs to be observed
by the CR for the estimation of $h_{cp}/\sigma_p^2$ in the case of
CP/TCI power control for the PR.

Note that the proposed new method for the CR to estimate $h_{cp}$
works in both cases of TDD and FDD modes for the PR. For comparison,
consider the conventional method where CR-Tx estimates $h_{cp}$ from
the received signal power from PR-Rx (when it transmits), denoted by
$\hat{q}_p=g_{pc}\hat{p}_p$, with $g_{pc}$ denoting the channel
power gain from PR-Rx to CR-Tx and $\hat{p}_p$ denoting the
instantaneous transmit power of PR-Rx. In contrast, the proposed
method estimates $h_{cp}$ at either CR-Tx or CR-Rx based on the
received signals from PR-Tx. There are three major advantages of the
proposed method over the conventional method. First, for the
conventional method, even in the case of PR TDD mode where channel
reciprocity holds such that $g_{pc}=h_{cp}$, $h_{cp}$ can be
estimated only if $\hat{p}_p$ is known at CR-Tx, which may not hold
in practice. In contrast, from (\ref{eq:channel gain estimate}) it
is observed that the proposed method does not rely on the knowledge
of PR transmit power. Second, the assumption $g_{pc}=h_{cp}$ for the
conventional method becomes problematic if FDD mode is used for the
PR, since $g_{pc}$ and $h_{cp}$ now correspond to two different
frequency bands and are thus different in general. In contrast, the
proposed method works independent of the relationship between
$g_{pc}$ and $h_{cp}$. Third, the conventional method may estimate
$h_{cp}$ but cannot give any information on the noise power at
PR-Rx, $\sigma_p^2$; as a result, CR-Tx cannot predict its resulting
interference power level at PR-Rx relative to $\sigma_p^2$. In
contrast, the proposed method provides the direct estimate on
$h_{cp}/\sigma_p^2$.

\subsection{Implementation}\label{subsec:implementation}

In this subsection, we address various implementation issues for the
proposed active learning scheme. First, we present the transmission
protocols for the PR and CR as follows.

\begin{itemize}
\item {\it PR Transmission Protocol}: We consider the conventional pilot-training-based transmission
protocol for the PR, where the transmission of PR-Tx is divided into
orthogonal time blocks, each of which is further divided into two
sub-blocks: one contains the training signal and the other contains
the data signal, as shown in Fig. \ref{fig:protocol}(a). The
training signal is for PR-Rx to estimate the PR channel $h_p$ as
well as the received noise power $N_p=\sigma_p^2+h_{cp}p_c$
(including the received CR interference power if $p_c>0$). It is
assumed that these estimates are perfect since in this paper we
focus on the deign of CR transmission. Based on the estimated $h_p$
and $N_p$, PR-Rx computes the effective channel power gain
$\gamma_p=h_p/N_p$, and according to $\gamma_p$ designs a feedback
signal for PR-Tx to adapt its transmit power and/or rate for the
next block transmission (for simplicity, we assume that there is no
delay or error for the PR feedback).

\item {\it CR Transmission Protocol}: As shown in Fig.
\ref{fig:protocol}(b), the transmission protocol for the CR is more
sophisticated than the conventional pilot-training-based one for the
PR. Specifically, each CR block transmission consists of four
stages: initial sensing, probing, re-sensing, and data transmission.
For initial sensing, CR-Tx observes the received PR signal power
$q_p^{(0)}$ and/or rate $r_p^{(0)}$. Then, in the probing stage,
CR-Tx transmits a predesigned signal of power $p_c$ to interfere
with PR-Rx. The probing signal of CR-Tx can also be used as the
training signal for CR-Rx. After that, CR-Tx goes into the
re-sensing stage to observe the updated PR signal power $q_p^{(1)}$
and/or rate $r_p^{(1)}$, and estimates $h_{cp}/\sigma_p^2$ according
to the rule given in (\ref{eq:channel gain estimate}). Last, based
on the estimated channel and the observed PR transmit adaptations,
CR-Tx sets its transmit power and rate (details are given later in
Section \ref{sec:CR transmission}), and starts data transmission.
\end{itemize}

Next, we discuss the following important issues for implementing the
above CR transmission protocol based on active learning.

\subsubsection{Time Synchronization}

One important issue for the proposed scheme is the timing
discrepancy between the distributed PR and CR links due to the lack
of a common reference clock. Let $\tau_p$, $\tau_{pc}$, and
$\tau_{cp}$ denote the propagation delays from PR-Tx to PR-Rx, from
PR-Tx to CR-Tx, and from CR-Tx to PR-Rx, respectively, with
$\tau_p\leq (\tau_{pc}+\tau_{cp})$. In addition, let $s_p(t)$ denote
the transmitted signal from PR-Tx. Then, the received signals at
PR-Rx and CR-Tx are $s_p(t-\tau_p)$ and $s_p(t-\tau_{pc})$ (the
channel multiplicative effect is ignored here since it is irrelevant
to the discussion on time synchronization), respectively. Since
CR-Tx does not have a common clock with PR-Tx, it has to use the
received signal from PR-Tx as a reference clock. Hence, the
transmitted probing signal from CR-Tx can be denoted as
$s_c(t-\tau_{pc}+\Delta)$, where $\Delta>0$ denotes the transmission
time ahead of the reference clock (to be specified later).
Accordingly, the received probing signal at PR-Rx is
$s_c(t-\tau_{pc}+\Delta-\tau_{cp})$. Note that CR-Tx needs to make
sure that its probing signal arrives at PR-Rx prior to the PR
training signal in one particular transmission block, i.e.,
$\tau_{pc}-\Delta+\tau_{cp}\leq \tau_p$, to make an effective
probing. Thus, it follows that $\Delta\geq
\tau_{pc}+\tau_{cp}-\tau_p>0$. However, the exact values of
$\tau_p$, $\tau_{pc}$, and $\tau_{cp}$ may not be known to CR-Tx.
Instead, suppose that we know that the maximum propagation delay
between CR and PR terminals is less than $\tau_{\max}$. Then, by
setting $\Delta=2\tau_{\max}$, it is ensured that the CR probing
signal arrives at PR-Rx prior to the PR training signal.

On the other hand, the duration of the probing signal from CR-Tx,
denoted by $T_c$, also needs to be properly designed. Note that in
order to minimize the temporary performance degradation of the PR
link due to the CR probing signal, it is desirable to choose a small
value for $T_c$. However, for the probing signal to be effective, it
is also necessary to make $T_c$ sufficiently large such that the
probing signal can overlap with the entire training signal of the PR
at PR-Rx in one particular transmission block. Let $T_p$ denote the
training signal duration of the PR, which is assumed known at CR-Tx.
From the earlier discussion on time synchronization, we know that
PR-Rx observes the PR signal, $s_p(t-\tau_{p})$, and CR probing
signal, $s_c(t-\tau_{pc}+2\tau_{\max}-\tau_{cp})$. Thus, the maximal
gap for the arrival time of the CR probing signal ahead of that of
the PR training signal is $2\tau_{\max}$ when
$\tau_p=(\tau_{pc}+\tau_{cp})$. Therefore, by setting
$T_c=T_p+2\tau_{\max}$, the aforementioned requirements for choosing
$T_c$ are both fulfilled.

\subsubsection{Rate Granularity}

In the estimation rule given by (\ref{eq:channel gain estimate}), it
has been assumed that the transmit rate of the PR,
$\mathcal{R}_p(SNR_p)$, is a continuous function of receiver SNR,
$SNR_p$. However, with practical MCSs, $\mathcal{R}_p(SNR_p)$ is
usually a non-decreasing function of $SNR_p$ with a finite rate
granularity, i.e., constituting only a finite number of discrete
rate values. In this case, suppose that
$\mathcal{R}_p(SNR_p^{(i)})=r_p^{(i)}$, with $0<SNR_L^{(i)}\leq
SNR_p^{(i)}< SNR_U^{(i)}$, $i=0,1$, where $r_p^{(i)}$ denotes a
discrete rate value, and $SNR_L^{(i)}$ and $SNR_U^{(i)}$ are
corresponding SNR thresholds. In this case, although the CR cannot
determine the exact value of $h_{cp}/\sigma_p^2$ from
(\ref{eq:channel gain estimate}), it can safely estimate the range
of this value as
\begin{equation}\label{eq:channel gain estimate bounds}
\left(\frac{SNR_L^{(0)}q_p^{(1)}}{SNR_U^{(1)}q_p^{(0)}}-1\right)\frac{1}{p_c}\leq
\frac{h_{cp}}{\sigma_p^2}\leq
\left(\frac{SNR_U^{(0)}q_p^{(1)}}{SNR_L^{(1)}q_p^{(0)}}-1\right)\frac{1}{p_c}.
\end{equation}

\subsubsection{Power Measurement Noise}

Another assumption we have made on the estimation using
(\ref{eq:channel gain estimate}) is that the sensor noise at CR-Tx
is ignored for estimating the received PR signal powers, $q_p^{(0)}$
and $q_p^{(1)}$, before and after the CR probing. In practice, only
a finite number of PR signal samples can be obtained during the
initial sensing and re-sensing periods at CR-Tx, which are corrupted
by the receiver noise. For convenience, we assume that the noise
power at CR-Tx is $\sigma_c^2$, the same as that at CR-Rx, and
$\sigma_c^2$ is known to CR-Tx. Also assume that $M$ independent
signal samples are obtained during both the initial sensing and
re-sensing periods at CR-Tx, denoted by
$\tilde{s}_p^{(i)}(1),\ldots,\tilde{s}_p^{(i)}(M)$, $i=0,1$.
Specifically, we have
\begin{equation}
\tilde{s}_p^{(i)}(m)=s_p^{(i)}(m)+\nu^{(i)}(m), ~ m=1,\ldots,M
\end{equation}
where $s_p^{(i)}(m)$ denotes the PR signal component, with
$\frac{1}{M}\sum_{m=1}^M|s_p^{(i)}(m)|^2\simeq q_p^{(i)}, i=0,1$,
and $\nu^{(i)}(m)$'s are independent Gaussian noises with zero mean
and variance of $\sigma_c^2$. Instead of having the exact values for
$q_p^{(0)}$ and $q_p^{(1)}$, we can obtain their estimated values as
follows.
\begin{equation}
\hat{q}_p^{(i)}=\frac{1}{M}\sum_{m=1}^M|\tilde{s}_p^{(i)}(m)|^2-\sigma_c^2,
~i=0,1.
\end{equation}

According to the central limit theorem \cite{Papoulis}, if the
number of samples $M$ is large enough (e.g., $\geq 10$ in practice),
the above estimation statistics are asymptotically normally
distributed with corresponding mean
\begin{equation}
\mathtt{E}(\hat{q}_p^{(i)})=q_p^{(i)}, ~i=0,1
\end{equation}
and variance
\begin{equation}
c^{(i)}:=\mathtt{Var}(\hat{q}_p^{(i)})=\frac{2\sigma_c^2(\sigma_c^2+2q_p^{(i)})}{M},
~i=0,1.
\end{equation}
Since $q_p^{(i)}$'s are unknown at CR-Tx, the exact values of
$c^{(i)}$'s are not available at CR-Tx. However, if it is known that
the PR transmit powers must be below a prescribed maximum value,
denoted by $P_{\max}$, the upper bounds for $c^{(i)}$'s can be
obtained as
\begin{equation}\label{eq:define c}
c^{(i)}\leq \frac{2\sigma_c^2(\sigma_c^2+2P_{\max})}{M}:=\hat{c},
~i=0,1.
\end{equation}
Thus, it follows that
\begin{align}
\mathtt{Prob}\left(\hat{q}_p^{(1)}\leq \left(q_p^{(1)}-
\zeta\sqrt{\hat{c}}\right)\right)\leq&~
\mathtt{Prob}\left(\hat{q}_p^{(1)}\leq \left(q_p^{(1)}-
\zeta\sqrt{c^{(1)}}\right)\right)\\ =&~ Q(\zeta)
\end{align}
where $Q(\cdot)$ is the complementary cumulative distribution
function \cite{Papoulis}, and $\zeta>0$ is a design parameter.
Similarly, we have
\begin{align}
\mathtt{Prob}\left(\hat{q}_p^{(0)}\geq \left(q_p^{(0)}+
\zeta\sqrt{\hat{c}}\right)\right)\leq Q(\zeta).
\end{align}
In other words, we have a belief in probability of at least
$1-Q(\zeta)$ for $\hat{q}_p^{(1)}> \left(q_p^{(1)}-
\zeta\sqrt{\hat{c}}\right)$ and $\hat{q}_p^{(0)}<
\left(q_p^{(0)}+\zeta\sqrt{\hat{c}}\right)$. Accordingly, from
(\ref{eq:channel gain estimate}), it follows that with a probability
of at least $1-Q(\zeta)$
\begin{equation}\label{eq:channel gain estimate new UB}
\frac{h_{cp}}{\sigma_p^2}\leq\left(\frac{\mathcal{R}^{-1}_p(r_p^{(0)})\left(\hat{q}_p^{(1)}+\zeta\sqrt{\hat{c}}\right)}
{\mathcal{R}^{-1}_p(r_p^{(1)})\left(\hat{q}_p^{(0)}-\zeta\sqrt{\hat{c}}\right)}-1\right)
\frac{1}{p_c}.
\end{equation}
Similarly, with the same probability guarantee, we have
\begin{equation}\label{eq:channel gain estimate new LB}
\frac{h_{cp}}{\sigma_p^2}\geq\left(\frac{\mathcal{R}^{-1}_p(r_p^{(0)})\left(\hat{q}_p^{(1)}-\zeta\sqrt{\hat{c}}\right)}
{\mathcal{R}^{-1}_p(r_p^{(1)})\left(\hat{q}_p^{(0)}+\zeta\sqrt{\hat{c}}\right)}-1\right)
\frac{1}{p_c}.
\end{equation}
Note that in (\ref{eq:channel gain estimate new UB}) and
(\ref{eq:channel gain estimate new LB}), we have assumed that
$\hat{q}_p^{(0)}>\zeta\sqrt{\hat{c}}$ and
$\hat{q}_p^{(1)}>\zeta\sqrt{\hat{c}}$, respectively. Thus, even with
a finite number of observation samples corrupted by additive noises,
CR-Tx can still obtain a pair of upper and lower bounds on
$h_{cp}/\sigma_p^2$ with a large belief probability (by setting a
sufficiently large value for $\zeta$). However, if the chosen
$\zeta$ is too large, it also increases the uncertainty range for
the estimation.

\subsubsection{Channel Variation}
Last, we address the issue on possible channel variations during the
implementation of the proposed CR active learning scheme. It is
worth noting that the assumption of constant channels has usually
been made in prior works (see, e.g., \cite{Foschini,Yates,Yu02}) on
iterative user power/rate adaptations in decentralized multiuser
systems. From the proof of Proposition \ref{proposition:channel
estimation}, we see that if the channel power gain,
$\tilde{h}_{pc}$, through which CR-Tx estimates the received signal
powers $q_p^{(0)}$ and $q_p^{(1)}$ from PR-Tx, changes from the
initial sensing stage to the re-sensing stage, the estimation result
will get affected. Let $\tilde{h}_{pc}^{(0)}$ and
$\tilde{h}_{pc}^{(1)}$ denote the true values of $\tilde{h}_{pc}$
during the initial sensing and re-sensing periods, respectively. We
can rewrite the estimation rule in (\ref{eq:channel gain estimate})
as (assuming the perfect rate and power estimation)
\begin{equation}\label{eq:estimate channel variation}
\frac{h_{cp}}{\sigma_p^2}=\left(\frac{\mathcal{R}^{-1}_p(r_p^{(0)})q_p^{(1)}\tilde{h}_{pc}^{(0)}}
{\mathcal{R}^{-1}_p(r_p^{(1)})q_p^{(0)}\tilde{h}_{pc}^{(1)}}-1\right)
\frac{1}{p_c}.
\end{equation}
Although CR-Tx does not know the exact values of
$\tilde{h}_{pc}^{(0)}$ and $\tilde{h}_{pc}^{(1)}$, it can predict
the approximate range for their ratio given the channel coherence
time relative to the time interval between the initial sensing and
re-seining stages, and obtain the corresponding upper and lower
bounds on the estimated value from (\ref{eq:estimate channel
variation}). Furthermore, the channel power gain $h_{cp}$ from CR-Tx
to PR-Rx may also change from the probing stage to the data
transmission stage. Similarly as for $\tilde{h}_{pc}$, given the
channel coherence time and the time interval between these two
stages, CR-Tx can estimate the range of $h_{cp}$ accordingly.

\section{Supervised Transmission} \label{sec:CR transmission}

In the previous section, we have proposed an active learning scheme
for the CR to estimate the channel gain from CR-Tx to PR-Rx by
exploiting the hidden PR feedback. In this section, we design
supervised transmission for CR data transmission stage shown in Fig.
\ref{fig:protocol}(b), based on the acquired knowledge from active
learning. In the following, we address two main design objectives
for CR supervised transmission: controlling the PR link performance
degradation and maximizing the CR link throughput.

\subsection{PR Performance Loss Control}\label{subsec:PR loss}

In this subsection, we illustrate how to apply the estimated
CR-to-PR channel gain from active learning for CR-Tx to predict the
performance loss of the PR link due to CR data transmission. For
simplicity, we assume that the estimation of $h_{cp}/\sigma_p^2$ is
perfect at CR-Tx, although the obtained results can be easily
extended to the case of imperfect channel estimation by utilizing
the derived estimation bounds in Section
\ref{subsec:implementation}. We consider two general types of
performance losses for the PR link: One is for the case where the PR
employs variable-rate transmission (e.g., with CP or WF power
control), named as {\it rate penalty}, which measures the PR rate
loss due to the CR interference, expressed as
$R_{l}=r_p^{(0)}-r_p^{(d)}$, where $r_p^{(d)}$ denotes the resultant
PR transmit rate in the CR data transmission stage; the other is for
the case where the PR employs constant-rate transmission (e.g., with
TCI power control), named as {\it power penalty}, which measures the
additional transmit power in dB required for the PR to maintain the
prescribed constant rate $r_p^{(0)}$ under the CR interference,
expressed as $P_l=10\times\log_{10}(p_p^{(d)}/p_p^{(0)})$, where
$p_p^{(d)}$ denotes the resultant PR transmit power in the CR data
transmission stage. Note that $r_p^{(0)}$ and $p_p^{(0)}$ denote the
PR transmit rate and power without the CR interference,
respectively, in the CR initial sensing stage. Let $p_c^{(d)}$
denote the CR transmit power in the data transmission stage.

First, the rate penalty for the PR link can be more explicitly
expressed as
\begin{align}\label{eq:rate loss}
R_{l}=\log_2\left(1+\frac{h_p
p_p^{(0)}}{\Gamma_p\sigma_p^2}\right)-\log_2\left(1+\frac{h_p
p_p^{(d)}}{\Gamma_p(\sigma_p^2+h_{cp}p_c^{(d)})}\right).
\end{align}
Note that for the convenience of analysis, we have assumed the ``SNR
gap approximation'' that accounts for the rate loss from the optimal
capacity due to practical/non-Gaussian MCS employed by the PR
\cite{Cioffi}, i.e.,
$\mathcal{R}_p(SNR_p)=\log_2(1+SNR_p/\Gamma_p)$, where $\Gamma_p\geq
1$ denotes the constant SNR gap for the PR.

In the case of CP policy for the PR, from (\ref{eq:rate loss}) it
follows that
\begin{align}
R_{l}^{\rm CP}&=\log_2\left(1+\frac{h_pQ}{\Gamma_p\sigma_p^2}\right)-
\log_2\left(1+\frac{h_pQ}{\Gamma_p(\sigma_p^2+h_{cp}p_c^{(d)})}\right) \\
&\leq \log_2\left(1+\frac{h_pQ}{\Gamma_p\sigma_p^2}\right)-\log_2\left(\frac{1+\frac{h_pQ}{\Gamma_p\sigma_p^2}}{1+\frac{h_{cp}p_c^{(d)}}{\sigma_p^2}}\right) \\
&= \log_2\left(1+\frac{h_{cp}p_c^{(d)}}{\sigma_p^2}\right)
\label{eq:rate loss UB}.
\end{align}
Therefore, CR-Tx knows that if it transmits with power $p_c^{(d)}$,
the resultant rate loss of the PR is upper-bounded by the value
given in (\ref{eq:rate loss UB}), which depends on the estimated
$h_{cp}/\sigma_p^2$, but is independent of the PR transmit power $Q$
and SNR gap $\Gamma_p$.

Consider next the case of WF power control for the PR similarly as
that given in (\ref{eq:WF}) but with $\gamma_p$ therein replaced by
$\gamma_p/\Gamma_p$. In this case, assuming that $r_p^{(0)}>0$
(otherwise the rate penalty for the PR is trivially zero), from
(\ref{eq:rate loss}) $R_{l}$ can be further expressed as
\begin{align}
R_{l}^{\rm WF}=\log_2\left(\frac{\mu
h_p}{\Gamma_p\sigma_p^2}\right)-\left(\log_2\left(\frac{\mu
h_p}{\Gamma_p(\sigma_p^2+h_{cp}p_c^{(d)})}\right)\right)^+.
\end{align}
It thus follows that
\begin{eqnarray}
R_{l}^{\rm
WF}=\left\{\begin{array}{ll}\log_2\left(1+\frac{h_{cp}p_c^{(d)}}{\sigma_p^2}\right)
& ~~ {\rm if}~
p_c^{(d)}\leq\frac{\frac{\mu h_p}{\Gamma_p\sigma_p^2}-1}{\frac{h_{cp}}{\sigma_p^2}}=\frac{2^{r_p^{(0)}}-1}{\frac{h_{cp}}{\sigma_p^2}} \\
r_p^{(0)} & ~~ {\rm otherwise} \end{array}\right..
\end{eqnarray}
Thus, CR-Tx can predict the exact rate loss of the PR as a function
of $p_c^{(d)}$, based on the estimated $h_{cp}/\sigma_p^2$ and
$r_p^{(0)}$ from the active learning.

Last, consider the power penalty of the PR with the TCI power
control given in (\ref{eq:TCI}). Assuming that
$r_p^{(d)}=r_p^{(0)}>0$, i.e., the CR interference power is not
sufficiently large to render the PR into a transmit outage
(otherwise the power penalty of the PR becomes irrelevant), it thus
follows that
\begin{equation}\label{eq:power penalty}
P_l^{\rm
TCI}=10\times\log_{10}\left(1+\frac{h_{cp}p_c^{(d)}}{\sigma_p^2}\right).
\end{equation}
Thus, CR-Tx can measure the power penalty of the PR as a function of
$p_c^{(d)}$.

From the above discussions, we see that the derived rate and power
penalties enable CR-Tx to predict quantitatively the resultant PR
performance losses corresponding to different transmit power levels
of the CR, using only the observed/estimated parameters from the
active learning.

\subsection{CR Achievable Rate}\label{subsec:CR rate}

In the previous subsection, we have shown for the CR supervised
transmission how to control the resultant PR link performance
degradation. With a given PR rate/power penalty, CR-Tx can derive
accordingly the maximum tolerable transmit power $p_c^{(d)}$. In
this subsection, we analyze the CR link achievable rate as a
function of $p_c^{(d)}$. Due to the space limitation, we consider
only the case of single-user detection at CR-Rx for decoding the CR
message, by treating the interference from PR-Tx as additive noise.
However, it is worth noting that more advanced multiuser detection
techniques can be employed at CR-Rx to decode both the CR and PR
messages in order to suppress the PR interference (details are
omitted here; the interested readers may refer to a preliminary
version of this paper \cite{Zhang09new}).

With single-user detection, the achievable rate of the CR link in
the data transmission stage can be expressed as
\begin{equation}\label{eq:rate SUD}
r_c^{(d)}=\log_2\left(1+\frac{h_cp_c^{(d)}}{\Gamma_c\left(\sigma_c^2+h_{pc}p_p^{(d)}\right)}\right)
\end{equation}
where $\Gamma_c\geq 1$ denotes the SNR gap for the CR, and
\begin{equation}
p_p^{(d)}=\mathcal{P}_p\left(\frac{h_p}{\sigma_p^2+h_{cp}p_c^{(d)}}\right)
\end{equation}
with $\mathcal{P}_p$ denoting the PR employed power control policy
(e.g., CR, TCI, or WF). It is interesting to observe that in general
the CR achievable rate is related to the CR transmit power
$p_c^{(d)}$ not only through the direct link from CR-Tx to CR-Rx,
but also through the interference link from CR-Tx to PR-Rx, the
resultant PR power adaptation and ``feedback'' interference from
PR-Tx to CR-Rx. Thus, CR-Tx is able to control the interference
power from PR-Tx by changing transmit power $p_c^{(d)}$ via the
hidden PR feedback.

With the PR feedback interference, some interesting observations can
be drawn for the CR achievable rate as a function of $p_c^{(d)}$.
Note that without the PR interference, $r_c^{(d)}$ is an increasing
function of $p_c^{(d)}$. However, with the PR feedback interference,
the interference power from PR-Tx can also be an increasing function
of $p_c^{(d)}$ in the case of persistent power control for the PR
(e.g., TCI). As a result, it is unclear in this case whether
increasing the CR transmit power will result in a net gain for its
achievable rate. Thus, it is pertinent to investigate further on
$r_c^{(d)}$ for the CR link under the PR feedback interference, as
shown in the following proposition.

\begin{proposition}\label{proposition:SUD}
For any $p_c^{(d)}\geq 0$ under which $\mathcal{P}_p(\gamma_p)$ with
$\gamma_p=\frac{h_p}{\sigma_p^2+h_{cp}p_c^{(d)}}$ is a positive,
continuous and differentiable function of $\gamma_p$,
$\frac{\partial r_c^{(d)}}{\partial p_c^{(d)}}>0$ if and only if
$\frac{\partial F(p_c^{(d)})}{\partial p_c^{(d)}}>0$, where
\begin{equation}
F(p_c^{(d)}):=\frac{p_c^{(d)}}{\sigma_c^2+h_{pc}\mathcal{P}_p\left(\frac{h_p}{\sigma_p^2+h_{cp}p_c^{(d)}}\right)}.
\end{equation}
\end{proposition}

The proof of Proposition \ref{proposition:SUD} follows from
(\ref{eq:rate SUD}) and is thus omitted here for brevity. It is
noted that CP and WF power control policies for the PR satisfy the
condition given in Proposition \ref{proposition:SUD}
straightforwardly, since they are both non-persistent power control.
For the TCI power control of the CR which is persistent, it can be
verified (details are omitted here for brevity) that $\frac{\partial
F(p_c^{(d)})}{\partial p_c^{(d)}}>0$, for all values of
$p_c^{(d)}\geq 0$ as required in Proposition \ref{proposition:SUD}.
It thus follows that $r_c^{(d)}$ is a strictly increasing function
of $p_c^{(d)}$ in all cases of CP, WF, or TCI power control policies
for the PR.

\section{Numerical Examples}\label{sec:numerical examples}

In this section, we present numerical examples to validate the
effectiveness of our proposed schemes for CR active learning and
supervised transmission. It is assumed that
$h_p=h_c=\tilde{h}_{pc}=1$ and $h_{cp}=h_{pc}=0.5$ in Fig.
\ref{fig:system model}. For simplicity, we assume that all these
channels are constant over the PR and CR transmission blocks where
the proposed CR schemes are implemented. We evaluate the performance
for the CR-to-PR channel gain estimation based on active learning,
as well as the PR performance degradation control and CR achievable
rate with CR supervised transmission. We consider the following two
scenarios: {\it Case I}: the PR employs a constant-power (CP)
variable-rate transmission; and {\it Case II}: the PR employs a
constant-rate variable-power (with TCI power control) transmission.
For convenience, we assume that $\sigma_p^2=\sigma_c^2=1$, and
$\Gamma_c=1$.

Consider first Case I, where the PR transmits with a constant power
$Q=100$. In this case, we are interested in investigating the
effects of finite rate granularity for the PR variable-rate
transmission on the performances of the CR active learning and
supervised transmission. Suppose that the PR transmit rate for a
given effective channel gain $\gamma_p$ is expressed as
\begin{equation}\label{eq:PR rate example}
r_p=\left\lfloor
\log_2\left(1+\frac{\gamma_pQ}{\Gamma_p}\right)\cdot\frac{1}{b}
\right \rfloor \cdot b
\end{equation}
in bps/Hz, where $\lfloor \cdot \rfloor$ denotes the floor
operation; and $b>0$ denotes the ``bit granularity'' due to the fact
that practical MCS only supports a finite set of discrete transmit
rates corresponding to integer multiplications of $b$. We assume
that $\Gamma_p=3$dB and $b=1$ (i.e., one-bit granularity). From
(\ref{eq:channel gain estimate bounds}), it follows that the upper
and lower bounds on $h_{cp}$ in the case of one-bit granularity are
obtained as
\begin{equation}\label{eq:active learning simulation}
\left(\frac{2^{r_p^{(0)}}-1}{2^{r_p^{(1)}+1}-1}-1\right)\frac{1}{p_c}\leq
h_{cp}\leq
\left(\frac{2^{r_p^{(0)}+1}-1}{2^{r_p^{(1)}}-1}-1\right)\frac{1}{p_c}
\end{equation}
where $r_p^{(0)}$ and $r_p^{(1)}$ denote the discrete rates of the
PR observed by the CR in the sensing and re-sensing stages,
respectively. In Fig. \ref{fig:PR CP}(a), we show the estimated
upper and lower bounds for $h_{cp}$  using the above estimation
rule. It is observed that with small value of CR probing signal
power $p_c$, the gap between the estimated upper and lower bounds
for $h_{cp}$ is large, suggesting that the estimation of $h_{cp}$ is
not accurate. This is due to the fact that if $p_c$ is too small,
the interference at PR-Rx is not sufficiently strong to make the PR
reduce its transmit rate by at least one bit (Note that $b=1$), and
as a result, the CR observes the same value of $r_p^{(1)}$ as
$r_p^{(0)}$. However, with larger value of $p_c$, the CR is able to
make $r_p^{(1)}< r_p^{(0)}$ and thus obtain a more accurate
estimation for $h_{cp}$. Thus, there is in general a tradeoff
between minimizing the PR performance degradation and the CR-to-PR
channel estimation error for the CR active learning. In Fig.
\ref{fig:PR CP}(b) and \ref{fig:PR CP}(c), we show the PR rate
penalty and CR achievable rate, respectively, vs. CR transmit power
$p_c^{(d)}$ for CR supervised data transmission. It is observed that
both the PR rate penalty and CR transmit rate increase with
$p_c^{(d)}$. Moreover, in Fig. \ref{fig:PR CP}(b), we compare the
actual resultant PR rate penalty (with one-bit granularity) to its
estimated value using (\ref{eq:rate loss UB}) and the estimated
upper bound on $h_{cp}$ from active learning with $p_c=10$. It is
observed that the estimated PR rate penalties are indeed valid upper
bounds on their true values for different values of $p_c^{(d)}$.

It is worth comparing the spectrum-sharing performance for the PR
and CR links with the proposed active learning and supervised
transmission for the CR, with the approach (refereed to as ``No
Feedback'') without exploiting the PR hidden feedback, or the
approach (refereed to as ``Perfect Feedback'') with the perfect
knowledge of the CR-to-PR channel via a dedicated feedback channel
from PR-Rx to CR-Tx. Note that for all three design approaches, the
achievable rates for the CR with a given transmit power $p_c^{(d)}$
are identical, as shown in Fig. \ref{fig:PR CP}(c). However, the
main differences among these designs are highlighted as follows. For
the case of ``No Feedback'', the CR has no means to predict the PR
performance loss as a function of $p_c^{(d)}$ and thus cannot deploy
any opportunistic transmission; as a result, the CR has to transmit
constantly with a very low power and thus results in low spectral
efficiency. In contrast, with the new proposed design, the CR can
always predict its maximum transmit power given the PR transmission
margin and decide its transmit rate accordingly. On the other hand,
for the case of ``Perfect Feedback'', as shown in Fig. \ref{fig:PR
CP}(b), for a given PR rate penalty value, the CR with the perfect
channel knowledge can transmit with a larger power than the proposed
design with active learning based channel estimation, and thus the
maximum achievable rate for the CR also becomes larger (cf. Fig.
\ref{fig:PR CP}(b) \& \ref{fig:PR CP}(c)).

Next, consider Case II, where the PR transmits with a constant rate
or equivalently maintains a constant receiver SNR, $SNR_p^{(T)}=10$.
Thus, the TCI power control given in (\ref{eq:TCI}) is used by the
PR with $\gamma_p^{(T)}=0.1$. In this case, we are interested in
investigating the effects of a finite number of observation samples
and receiver noise at CR-Tx for estimating the received PR signal
powers on the performances of CR active learning and supervised
transmission. From (\ref{eq:channel gain estimate new UB}) and
(\ref{eq:channel gain estimate new LB}), it follows that the upper
and lower bounds on $h_{cp}$ in the case of a finite number of
observed PR signal samples are obtained as
\begin{equation}\label{eq:active learning simulation 2}
\left(\frac{\left(\hat{q}_p^{(1)}-\zeta\sqrt{\hat{c}}\right)}
{\left(\hat{q}_p^{(0)}+\zeta\sqrt{\hat{c}}\right)}-1\right)
\frac{1}{p_c}\leq h_{cp}\leq
\left(\frac{\left(\hat{q}_p^{(1)}+\zeta\sqrt{\hat{c}}\right)}
{\left(\hat{q}_p^{(0)}-\zeta\sqrt{\hat{c}}\right)}-1\right)
\frac{1}{p_c}
\end{equation}
where $\hat{q}_p^{(0)}$ and $\hat{q}_p^{(1)}$ denote the observed
powers at CR-Tx in the sensing and re-sensing stages, respectively.
In order to keep the estimated $h_{cp}$ within the above range with
a probability guarantee of $99\%$, we choose $\zeta=2.3$ since
$Q(2.3)\approx0.01$. Furthermore, we set $P_{\max}=100$ and $M=500$
for determining the constant $\hat{c}$ defined in (\ref{eq:define
c}). In Fig. \ref{fig:PR TCI}(a), we show the estimated upper and
lower bounds for $h_{cp}$ using the above rule. Similar to our
previous observations for Fig. \ref{fig:PR CP}(a), it is observed
that the CR probing power $p_c$ needs to be sufficiently large in
order to make a reasonably good estimate on $h_{cp}$. In Fig.
\ref{fig:PR TCI}(b) and \ref{fig:PR TCI}(c), we show the PR power
penalty and CR achievable rate, respectively, vs. CR transmit power
$p_c^{(d)}$ for CR supervised data transmission. It is observed that
both the PR power penalty and CR transmit rate increase with
$p_c^{(d)}$. Moreover, in Fig. \ref{fig:PR TCI}(b), we compare the
actual PR power penalty to its estimated value using (\ref{eq:power
penalty}) and the estimated upper bound on $h_{cp}$ from active
learning with $p_c=10$. It is observed that the estimated PR power
penalties are valid upper bounds on the true values, which become
tighter for smaller values of $p_c^{(d)}$. Comparing the CR
achievable rates in Fig. \ref{fig:PR CP}(c) and Fig. \ref{fig:PR
TCI}(c), it is observed that the CR rate increase with $p_c^{(d)}$
is much slower in the latter than the former case. This is because
for Fig. \ref{fig:PR TCI}(c), the PR employs TCI power control
instead of CP as for Fig. \ref{fig:PR CP}(c), and thus the PR
feedback interference power at CR-Rx increases with $p_c^{(d)}$
instead of being a constant as for the case of Fig. \ref{fig:PR
CP}(c) with CP.

\section{Conclusion} \label{sec:conclusion}

This paper introduces a new design paradigm for spectrum sharing
based CRs, where the CR designs its learning and transmission from
the observed PR transmit power/rate adaptations upon receiving a
probing signal from the CR, namely the hidden PR feedback. First, a
novel active learning scheme is proposed for the CR to estimate the
channel gain from its transmitter to the PR receiver, which is
essential for the CR interference control to the PR. Second, with
the acquired channel knowledge and PR transmit adaptations from
active learning, the CR supervised data transmission is designed by
effectively controlling the performance degradation of the PR as a
function of the CR transmit power. Moreover, this paper shows that
the CR is able to predict its own achievable rate under the PR
feedback interference, which is coupled with the CR transmit power
via the hidden PR feedback. This paper presents a new transmission
protocol for the CR to implement the proposed learning and
transmission schemes, and proposes the solutions to deal with
various important practical issues. The results in this paper
provide a new promising approach to interference management for
decentralized multiuser communication systems.

\newpage

\begin{figure}
\psfrag{a}{\hspace{0.3in} PR-Tx}\psfrag{b}{\hspace{0.3in} CR-Tx}
\psfrag{c}{PR-Rx}\psfrag{d}{CR-Rx}\psfrag{e}{$h_p$}\psfrag{f}{$h_c$}
\psfrag{g}{\hspace{-0.1in}$h_{pc}$}\psfrag{h}{$h_{cp}$}\psfrag{i}{\hspace{-0.05in}$\tilde{h}_{pc}$}
\begin{center}
\scalebox{1.0}{\includegraphics*[54pt,594pt][368pt,755pt]{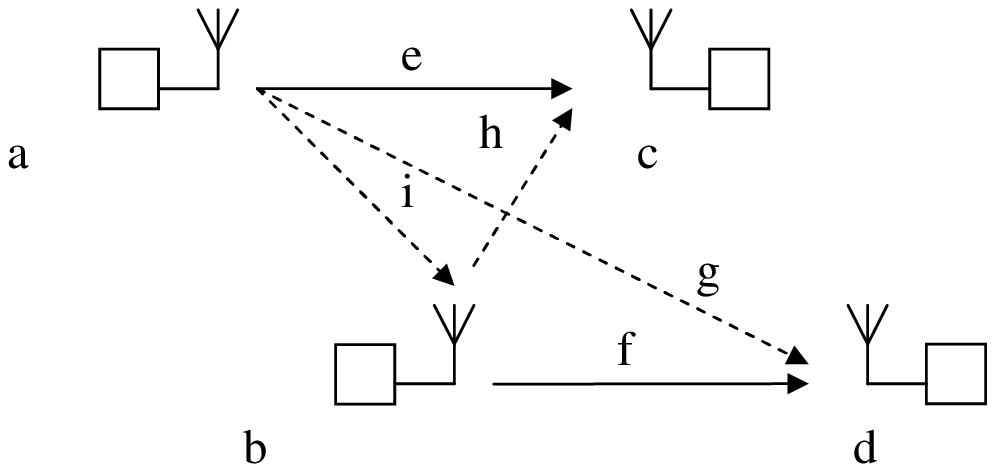}}
\end{center}
\caption{Spectrum sharing between a PR link and a CR
link.}\label{fig:system model}
\end{figure}

\begin{figure}
\centering{
 \epsfxsize=4.5in
    \leavevmode{\epsfbox{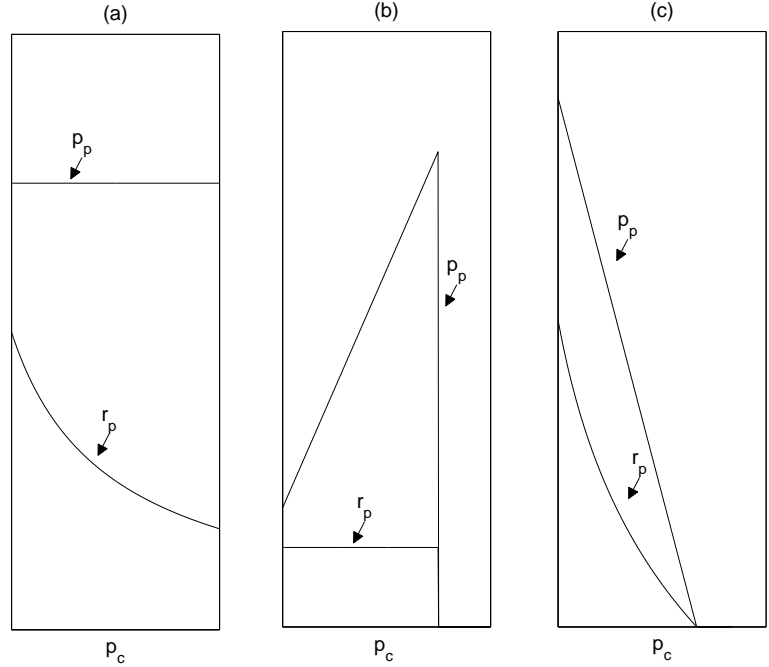}} }
\caption{Plots of $p_p$ and $r_p$ as functions of $p_c$ for (a) CP;
(b) TCI; and (c) WF power control of the PR.}\label{fig:PR rate
power}
\end{figure}

\begin{figure}
\psfrag{c}{Training}\psfrag{d}{Data Transmission}
\psfrag{e}{Sensing}\psfrag{f}{Probing}\psfrag{g}{Re-sensing}
\begin{center}
\scalebox{0.8}{\includegraphics*[44pt,634pt][370pt,740pt]{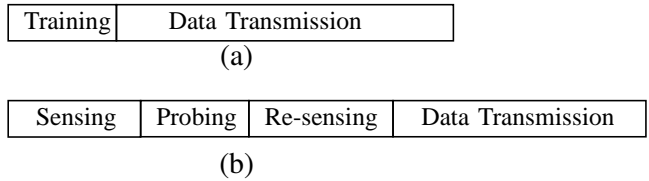}}
\end{center}
\caption{Transmission protocols for (a) the PR; and (b) the
CR.}\label{fig:protocol}
\end{figure}

\begin{figure}
\centering{
 \epsfxsize=4.0in
    \leavevmode{\epsfbox{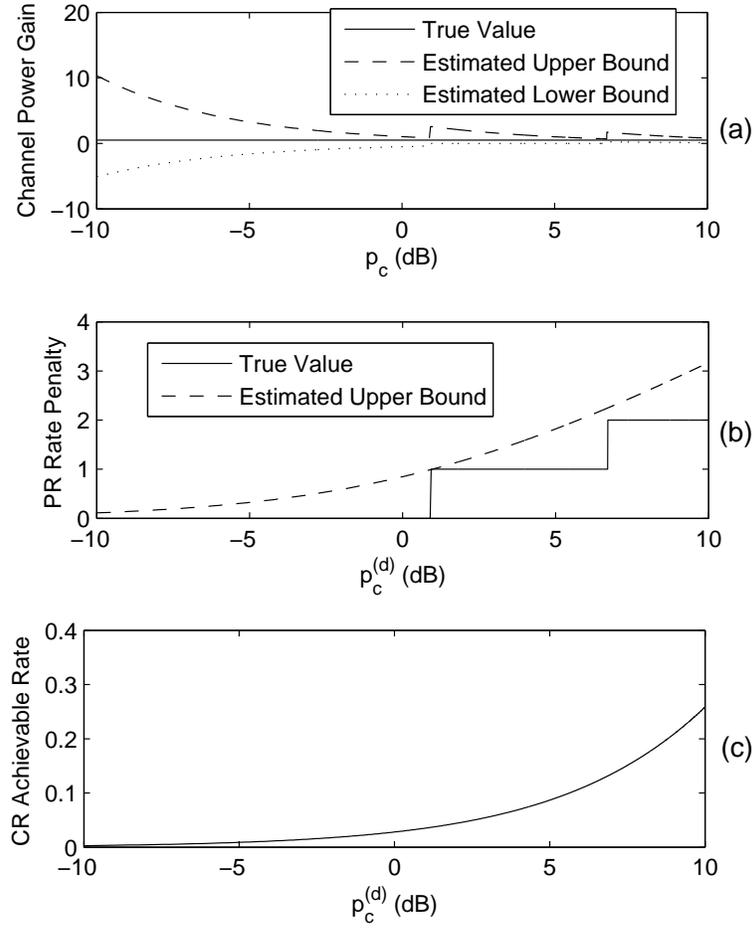}} }
\caption{Performance of CR active learning and supervised
transmission when PR employs constant-power variable-rate
transmission (Case I): (a) CR-to-PR channel power gain estimation;
(b) PR rate penalty; and (c) CR achievable rate.}\label{fig:PR CP}
\end{figure}

\begin{figure}
\centering{
 \epsfxsize=4.0in
    \leavevmode{\epsfbox{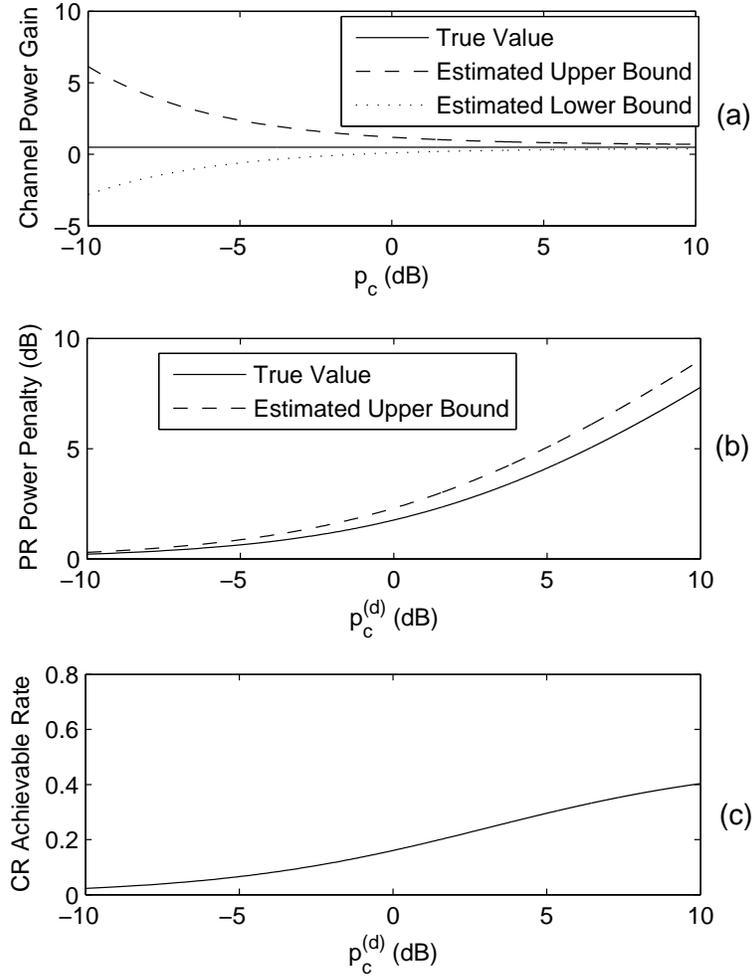}} }
\caption{Performance of CR active learning and supervised
transmission when PR employs constant-rate variable-power
transmission (Case II): (a) CR-to-PR channel power gain estimation;
(b) PR power penalty; and (c) CR achievable rate.}\label{fig:PR TCI}
\end{figure}

\end{document}